\documentclass[runningheads,orivec]{llncs}
\usepackage{hyperref}
\usepackage{verbatim}
\usepackage{url}
\usepackage{listings}
\usepackage[vlined,linesnumbered,ruled]{algorithm2e}
\usepackage[dvipsnames]{xcolor}
\usepackage{amsmath,amssymb,verbatim,mathtools}
\usepackage[capitalise]{cleveref}
\usepackage{xspace}
\usepackage[font=small,skip=0pt,labelsep=period,labelfont=bf, 
belowskip=-15pt]{caption}
\SetAlgoSkip{}
\SetAlgoInsideSkip{}

\newcommand{\dl}[1]{\ensuremath{\mathcal{#1}}\xspace}
\newcommand{\dlF}{\dl F}
\newcommand{\dlK}{\dl K}
\newcommand{\dlT}{\dl T}
\newcommand{\dlA}{\dl A}
\newcommand{\dlS}{\dl S}
\newcommand{\dlM}{\dl M}
\newcommand{\dlX}{\dl X}

\newcommand{\dlI}{\dl I}

\newcommand{\kb}{\ensuremath{\left< \dlT, \dlA \right>}\xspace}
\newcommand{\core}{\ensuremath{\left< \dlS, \dlM \right>}\xspace}
\newcommand{\ccKB}{\ensuremath{\dlK = \left< \dlT, \dlA, \dlS, \dlM 
\right>}\xspace}
\newcommand{\partccKB}{\ensuremath{\dlK = \left< \dlT, \dlA, \dlS, \partitionedM 
\right>}\xspace}
\newcommand{\dlLiteF}{DL-Lite$^\dlF$\xspace}

\newcommand{\dlLiteA}{DL-Lite$^\dlA$\xspace}
\newcommand{\ALCOIQ}{\ensuremath{\dl{ALCOIQ}}\xspace}
\newcommand{\partitionedM}{\ensuremath{\{\mathcal M_i\}_{i \in \declaredResourcesNodes}}}

\newcommand{\concepts}{\ensuremath{\mathbf{C}}\xspace}
\newcommand{\specConcepts}{\ensuremath{\concepts^\dlS}\xspace}
\newcommand{\openConcepts}{\ensuremath{\concepts^\dlK}\xspace}
\newcommand{\roles}{\ensuremath{\mathbf{R}}\xspace}
\newcommand{\specRoles}{\ensuremath{\roles^\dlS}\xspace}
\newcommand{\openRoles}{\ensuremath{\roles^\dlK}\xspace}
\newcommand{\individuals}{\ensuremath{\mathbf{I}}\xspace}
\newcommand{\modelNodes}{\ensuremath{\individuals^\dlM}\xspace}
\newcommand{\boundaryNodes}{\ensuremath{\individuals^B}\xspace}
\newcommand{\openIndividuals}{\ensuremath{\individuals^\dlK}\xspace}
\newcommand{\declaredResourcesNodes}{\ensuremath{\individuals^{r}}\xspace}
\newcommand{\genericAlphabet}{\ensuremath{\mathbf{X}}\xspace}

\title{Actions over Core-closed Knowledge Bases}

\author{ Claudia Cauli \inst{1,2}\thanks{This 
work was done prior to joining Amazon.}
	\and Magdalena Ortiz \inst{3}
	\and Nir Piterman \inst{1}\thanks{Supported by ERC Consolidator grant D-SynMA (No. 772459).}}

\institute{University of Gothenburg \and Amazon 
Web Services \and TU Wien }

\begin{document}
\maketitle

\begin{abstract}
	We present new results on the application of semantic- and 
knowledge-based reasoning techniques to the analysis of cloud 
deployments. 
In particular, to the security of \emph{Infrastructure as Code} 
configuration files, encoded as description logic knowledge bases. We 
introduce an action language to model \emph{mutating actions}; that 
is, actions that change the structural configuration of a given 
deployment by adding, modifying, or deleting resources. 
We mainly focus on two problems: the problem of determining whether 
the execution of an action, no matter the parameters passed to it, 
will not cause the violation of some security requirement 
(\emph{static verification}), and the problem of finding sequences of 
actions that would lead the deployment to a state where (un)desirable 
properties are (not) satisfied (\emph{plan existence} and \emph{plan 
synthesis}). For all these problems, we provide definitions, 
complexity results, and decision procedures.
\end{abstract}

\section{Introduction}

The use of automated reasoning techniques to analyze properties of 
cloud infrastructure is gaining increasing attention 
\cite{BBCDGLRTV18,BBCDGHKKKK19,BCCDGHJMP20,BBBBCGJLMM20,Cook18}.
Despite that, more effort needs to be put into the modeling and 
verification of generic security requirements over cloud 
infrastructure pre-deployment. The availability of formal techniques, 
providing strong security guarantees, would assist complex 
system-level analyses such as threat modeling and data flow, 
which now require considerable time, manual intervention, and expert 
domain knowledge.

We continue our research on the application of semantic-based and  
knowledge-based reasoning techniques to cloud deployment 
\emph{Infrastructure as Code} configuration files. 
In~\cite{CauliLPT21}, we reported on our experience using expressive 
description logics to model and reason about Amazon Web Services' 
proprietary Infrastructure as Code framework (AWS CloudFormation). We 
used the rich constructs of these logics to encode domain knowledge, 
simulate closed-world reasoning, and express mitigations and 
exposures to security threats. Due to the high complexity of basic 
tasks~\cite{Tob99,BHLS17}, we found reasoning in such a framework to 
be not efficient at cloud scale.
In~\cite{CauliOP21}, we introduced \emph{core-closed knowledge 
bases}{---}a lightweight description logic combining closed- and 
open-world reasoning, that is tailored to model cloud infrastructure 
and efficiently query its security properties.
Core-closed knowledge bases enable partially-closed predicates whose
interpretation is closed over a \emph{core} part of the knowledge base
but open elsewhere.
To encode potential exposure to security threats, we studied the  
query satisfiability problem and (together with the usual query 
entailment problem) applied it to a new class of conjunctive queries 
that we called \textsc{Must}/\textsc{May} queries.
We were able to answer such queries over core-closed knowledge bases
in \textsc{LogSpace} in data complexity and \textsc{NP} in combined
complexity, improving on the required \textsc{NExptime} complexity 
for satisfiability over \ALCOIQ (used in \cite{CauliLPT21}).

Here, we enhance the quality of the analyses done over pre-deployment 
artifacts, giving users and practitioners additional precise insights 
on the impact of potential changes, fixes, and general improvements 
to their cloud projects.
To do so, we enrich core-closed knowledge bases with the notion of 
\emph{core-completeness}, which is needed to ensure that updates are 
consistent. We define the syntax and semantics of an action language 
that is expressive enough to encode \emph{mutating} API calls, i.e., 
operations that change a cloud deployment configuration by creating, 
modifying, or deleting existing resources. As part of our effort to 
improve the quality of automated analysis, we also provide relevant 
reasoning tools to identify and predict the consequences of these 
changes. To this end, we consider procedures that determine whether 
the execution of a mutating operation always preserves given 
properties (\emph{static verification}); determine whether there 
exists a sequence of operations that would lead a deployment to a 
configuration meeting certain requirements (\emph{plan existence}); 
and find such sequences of operations (\emph{plan synthesis}).

The paper is organized as follows.
In~\cref{sec:background} we provide background on core-closed 
knowledge bases, conjunctive queries, and \textsc{Must}/\textsc{May} 
queries. 
In \cref{sec:coreComp} we motivate and introduce the notion of 
\emph{core-completeness}.
In~\cref{sec:act} we define the action language. 
In~\cref{sec:sv} we describe the static verification problem and 
characterize its complexity.
In~\cref{sec:ts} we address the planning problem and concentrate on 
the synthesis of minimal plans satisfying a given requirement 
expressed using \textsc{Must}/\textsc{May} queries. 
We discuss related works in~\cref{sec:rel} and conclude 
in~\cref{sec:conc}.

\section{Background}\label{sec:background}

Description logics (DLs) are a family of logics for encoding  
knowledge in terms of concepts, roles, and individuals; analogous to 
first-order logic unary predicates, binary predicates, and constants, 
respectively.
Standard DL knowledge bases (KBs) have a set of axioms, called 
\emph{TBox}, and a set of assertions, called \emph{ABox}. The TBox 
contains axioms that relate to concepts and roles. The ABox contains 
assertions that relate individuals to concepts and pairs of 
individuals to roles. KBs are usually interpreted under the 
open-world assumption, meaning that the asserted facts are not 
assumed to be complete.

\paragraph{Core-closed Knowledge Bases}
In~\cite{CauliOP21}, we introduced {core-closed 
knowledge bases} (ccKBs) as a suitable description logic formalism to 
encode cloud deployments. The main characteristic of ccKBs is to 
allow for a combination of open- and closed-world reasoning that 
ensures tractability.
A \dlLiteF ccKB is the tuple \ccKB built from the standard knowledge 
base \kb and the \emph{core} system \core. 
The former encodes incomplete terminological and assertional 
knowledge.
The latter is, in turn, 
composed of two parts: \dlS, containing axioms that encode the core 
structural specifications, and \dlM, containing positive 
concept and role assertions that encode the core configuration. 
Syntactically, \dlM is similar to an ABox but, semantically, is 
assumed to be complete with respect to the specifications 
in \dlS.
The ccKB \dlK is defined over the alphabets \concepts (of concepts), 
\roles (of roles), and \individuals (of individuals), all 
partitioned into an open subset and a partially-closed subset. 
That is, the set of concepts is partitioned into the open concepts 
\openConcepts and the closed (specification) concepts \specConcepts;
the set of roles is partitioned into open roles 
\openRoles and closed (specification) roles \specRoles; and 
the set of individuals is partitioned into open individuals 
\openIndividuals and closed (model) individuals \modelNodes.
We call \specConcepts and \specRoles core-closed predicates, or 
partially-closed predicates, as their 
extension is closed 
over the core domain \modelNodes and open otherwise. In contrast, we 
call \openConcepts 
and \openRoles open predicates.
The syntax of concept and role expressions in 
\dlLiteF~\cite{ArtaleCKZ09,CalvaneseGLLR07}
 is as follows: 
\[
\mathsf{B} ::= \bot \mid \mathsf{A} \mid \exists\mathsf{p}
\]
where \textsf{A} denotes a concept name and \textsf{p} is either a 
role name \textsf{r} or its inverse \textsf{r}$^-$. The syntax of 
axioms provides for the 
three following axioms: 
\[
\mathsf{B}^1 \sqsubseteq \mathsf{B}^2, \qquad
\mathsf{B}^1 \sqsubseteq \neg\mathsf{B}^2, \qquad
(\mathsf{funct}\ \mathsf{p}),
\]
respectively called: \emph{positive inclusion} axioms, \emph{negative 
inclusion} axioms, and \emph{functionality} axioms.
From now on, we denote symbols from the alphabet 
$\genericAlphabet^\dlX$ with the 
subscript \dlX, and symbols from the generic alphabet 
$\genericAlphabet$ with no 
subscript.
In core-closed knowledge bases, axioms and assertions fall into the 
scope of a different set depending on the predicates and individuals 
that they refer to, according to the set definitions below.
\begin{align*}
\dlM \subseteq&\ \{\ 
\mathsf A_\dlS(a_\dlM),\ 
\mathsf R_\dlS(a_\dlM,a),\ 
\mathsf R_\dlS(a,a_\dlM)\ \}\\
\dlA \subseteq&\ \{\ 
\mathsf A_\dlK(a_\dlK),\ 
\mathsf R_\dlK(a_\dlK,b_\dlK),\ 
\mathsf{A}_\dlS(a_\dlK),\  
\mathsf R_\dlS(a_\dlK, b_\dlK) \ 
\}\\
\dlS \subseteq&\ \{\ 
\mathsf B^1_\dlS\!\sqsubseteq\! \mathsf B^2_\dlS,\ 
\mathsf B^1_\dlS\! \sqsubseteq\! \neg\mathsf  
B^2_\dlS,\  
\mathsf{Func}(\mathsf P_\dlS)\ \}\\
\dlT \subseteq&\ \{\ 
\mathsf B^1\!\sqsubseteq\! \mathsf B^2_\dlK,\ 
\mathsf B^1\! \sqsubseteq\! \neg\mathsf  B^2_\dlK,\  
\mathsf{Func}(\mathsf P_\dlK) \ \}
\end{align*}
As mentioned earlier, \dlM-assertions are assumed to be complete and 
consistent with respect to the terminological knowledge given in 
\dlS; whereas the usual open-world assumption is made for 
\dlA-assertions.
The semantics of a \dlLiteF core-closed KB is given in terms of 
interpretations \dlI, consisting of a non-empty domain $\Delta^\dlI$ 
and an interpretation function $\cdot^\dlI$. The latter assigns to 
each concept \textsf{A} a subset $\mathsf A^\dlI$ of $\Delta^\dlI$, 
to each role \textsf{r} a subset $\mathsf r^\dlI$ of 
$\Delta^\dlI\times\Delta^\dlI$, and to each 
individual $a$ a node $a^\dlI$ in $\Delta^\dlI$, and it is extended 
to concept expressions in the usual way.
An interpretation \dlI is a model of an inclusion axiom $\mathsf
B_1\sqsubseteq \mathsf B_2$ if $\mathsf B_1^\dlI \subseteq \mathsf
B_2^\dlI$. 
An interpretation \dlI is a model of a membership assertion 
$\mathsf A(a)$, (resp. $\mathsf r(a,b)$) if 
$a^\dlI\!\in\!\mathsf A^\dlI$ (resp.  $(a^\dlI,b^\dlI)\!\in\!\mathsf
r^\dlI$).
We say that \dlI models \dlT, \dlS, and \dlA if it models all axioms 
or assertions contained therein.
We say that \dlI models \dlM, denoted 
$\dlI\models^{\mathsf{CWA}}\dlM$, when it models an 
\dlM-assertion $f$ \emph{if and only if } $f\!\in\!\dlM$.
Finally, \dlI models \dlK if it models \dlT, \dlS, \dlA, 
and \dlM. When \dlK has at least one model, we say that \dlK is 
satisfiable.

In the application presented in~\cite{CauliLPT21}, 
description logic KBs are used to encode machine-readable 
deployment files containing multiple resource declarations. 
Every resource declaration has an underlying tree structure, whose 
leaves can potentially link to the roots of other resource 
declarations. Let \declaredResourcesNodes$\subseteq$ \modelNodes be 
the set of all resource nodes, we encode their resource declarations 
in \dlM, and formalize the resulting forest structure by partitioning 
\dlM into multiple subsets \partitionedM, each representing a tree of 
assertions rooted at a resource node $i$ (we generally refer to
constants in $\dlM$ as nodes). For the purpose of this 
work, we will refer to core-closed knowledge bases where \dlM is 
partitioned as described; that is, ccKBs such that \partccKB.

\paragraph{Conjunctive Queries}
A \emph{conjunctive query} (CQ) is an existentially-quantified  
formula $q[\vec{x}]$ of the form $\exists \vec{y}.\textit{conj} 
(\vec{x},\vec{y})$, where \emph{conj} is a conjunction of positive 
atoms and potentially inequalities.
A \emph{union of conjunctive queries} (UCQ) is a disjunction of CQs. 
The variables in $\vec{x}$ are called \emph{answer variables}, those 
in $\vec{y}$ are the existentially-quantified \emph{query variables}.
A tuple $\vec{c}$ of constants appearing in the knowledge base \dlK 
is an answer to $q$ if for all interpretations \dlI model of \dlK we 
have $\dlI\models q[\vec{c}]$.
We call these tuples the \emph{certain answers} of \emph{q} over
\dlK, denoted $ans(\dlK,q)$, and the problem of testing whether a
tuple is a certain answer \emph{query entailment}.
A tuple $\vec{c}$ of constants appearing in \dlK satisfies $q$ if
there exists an interpretation \dlI model of \dlK such that
$\dlI\models q[\vec{c}]$.
We call these tuples the \emph{sat answers} of \emph{q} over \dlK,
denoted $sat{-}ans(\dlK,q)$, and the problem of testing whether a
given tuple is a sat answer \emph{query satisfiability}.

\paragraph{\textsc{Must}/\textsc{May} Queries}
A \textsc{Must}/\textsc{May} query $\psi$ (\cite{CauliOP21}) is a 
Boolean combination of 
nested UCQs in the scope of a \textsc{Must} or a \textsc{May} 
operator as follows: 
\[
\psi ::= 
\neg \psi\  \mid\ 
\psi_1 \wedge \psi_2 \mid\  
\psi_1 \vee \psi_2 \mid\  
\textsc{Must}\ \varphi \mid\  
\textsc{May}\ \varphi_{\not\approx}
\]
where $\varphi$ and $\varphi_{\not\approx}$ are unions of conjunctive 
queries potentially containing inequalities. The reasoning needed for 
answering the nested queries can be decoupled from the reasoning 
needed to answer the higher-level formula: nested queries 
$\textsc{Must}\ \varphi$ are reduced to conjunctive 
query entailment, and nested queries $\textsc{May}\ 
\varphi_{\not\approx}$ are reduced to 
conjunctive query satisfiability. We denote by $\mathsf{ANS}(\psi, 
\dlK)$ the answers of a \textsc{Must}/\textsc{May} query $\psi$ over 
the core-closed knowledge base \dlK.

\section{Core-complete Knowledge Bases}\label{sec:coreComp}

The algorithm \textsf{Consistent} presented 
in~\cite{CauliOP21} computes satisfiability of 
\dlLiteF core-closed knowledge bases relying on the assumption that 
\dlM is complete and consistent with respect to \dlS. Such an 
assumption 
effectively means that the information contained in \dlM is 
\emph{explicitly} present and \emph{cannot be completed by 
inference}. The algorithm relies on the existence of a theoretical 
object, the canonical interpretation, in which missing assertions  
can always be introduced when they are logically implied by the 
positive inclusion axioms. As a matter of fact, positive inclusion 
axioms are not even included in the inconsistency formula built for 
the satisfiability check, as it is proven that the canonical 
interpretation always satisfies them (\cite{CauliOP21}, 
Lemma 3). When the assumption that \dlM is consistent with respect to 
\dlS is dropped, the algorithm \textsf{Consistent} becomes 
insufficient to check satisfiability.
We illustrate this with an example.
\begin{example}[Required Configuration]\label{example:requiredConfig}
Let us consider the axioms constraining the AWS resource type 
$\mathsf{S3\!\!::\!\!Bucket}$. In particular, the 
\dlS-axiom $\mathsf{S3\!\!::\!\!Bucket} \sqsubseteq \exists 
\mathsf{loggingConfiguration}$ prescribing that all buckets must have 
a \emph{required} logging configuration. For a set $\dlM = \{
\mathsf{S3\!\!::\!\!Bucket}(b)\}$, according to the 
partially-closed semantics of core-closed knowledge bases, the 
absence of an assertion $\mathsf{loggingConfiguration}(b,x)$, for 
some $x$, is interpreted as the assertion being false in \dlM, which 
is therefore not consistent with respect to \dlS. 
However, the algorithm 
\textsf{Consistent} will check the \emph{lts} interpretation of \dlM 
for an empty formula (as there are no negative inclusion or 
functionality axioms) and return \emph{true}.
\end{example}
In essence, the algorithm \textsf{Consistent} does not compute the 
full satisfiability of the whole core-closed knowledge base, but only 
of its open part.
Satisfiability of \dlM with respect to the positive inclusion 
axioms in \dlS needs to be checked separately. 
We introduce a new notion to denote when a set \dlM is complete with 
respect to \dlS that is distinct from the notion of consistency. 
Let $\ccKB$ be a \dlLiteF core-closed knowledge base; we say that 
$\dlK$ is \emph{core-complete} when \dlM models \emph{all} positive 
inclusion axioms in \dlS under a closed-world assumption; we say that 
\dlK is \emph{open-consistent} when \dlM and \dlA model all negative 
inclusion and functionality axioms in \dlK's negative inclusion 
closure. 
Finally, we say that \dlK is 
\emph{fully satisfiable} when is both \emph{core-complete} and 
\emph{open-consistent}.

\begin{lemma}
In order to check \emph{full satisfiability} of a \dlLiteF 
core-closed KB, one simply needs to check if \dlK is 
\emph{core-complete} (that is, if \dlM models all \emph{positive 
axioms} in $\dlS$ under a closed-world assumption) and if \dlK is 
\emph{open-consistent} (that is, to run the algorithm 
\emph\textsf{Consistent}).
\end{lemma}

\begin{proof}
Dropping the assumption that \dlM is consistent w.r.t. \dlS causes 
Lemma 3 from~\cite{CauliOP21} to fail. In particular, 
the canonical interpretation of \dlK, $can(\dlK)$, would still be a 
model of $PI_\dlT$, $\dlA$, and $\dlM$, but may \emph{not} be a model 
of $PI_\dlS$. This is due to the construction of the canonical model 
that is based on the notion of applicable axioms. In rules 
\textbf{c5-c8} of~\cite{CauliOP21} Definition 1, axioms 
in $PI_\dlS$ are defined as applicable to assertions involving open 
nodes $a_\dlK$ but \emph{not} to model nodes $a_\dlM$ in 
$\modelNodes$. As a result, if the implications of such axioms on 
model nodes are not included in \dlM itself, then they will not be 
included in $can(\dlK)$ either, and $can(\dlK)$ will not be a model 
of $PI_\dlS$. On the other hand, one can easily verify that Lemmas 
1,2,4,5,6,7 and Corollary 1 would still hold as they do not rely on 
the assumption. However, since it is not guaranteed anymore that \dlM 
satisfies all positive inclusion axioms from \dlS, the \emph{if} 
direction of~\cite{CauliOP21} Theorem 1 does not hold 
anymore: there can be an unsatisfiable ccKB \dlK such that 
$db(\dlA)\cup lts(\dlM) \models cln(\dlT\cup\dlS),\dlA,\dlM$. For 
instance, the knowledge base from~\cref{example:requiredConfig}. We 
also note that the negative inclusion and functionality axioms from 
\dlS will be checked anyway by the consistency formula, both on 
$db(\dlA)$ and on $lts(\dlM)$.
\end{proof}

\begin{lemma}\label{lemma:fullSatComplexity}
Checking whether a \dlLiteF core-closed knowledge base is 
\emph{core-complete} can be done in polynomial time in \dlM. As a 
consequence, checking full satisfiability is also done in polynomial 
time in~\dlM.
\end{lemma}

\begin{proof}
One can write an algorithm that checks \emph{core-completeness} by 
searching for the existence of a 
positive inclusion axiom $\mathsf{B}^1_\dlS 
\sqsubseteq \mathsf{B}^2_\dlS \in PI_\dlS$ such that $\dlM\models 
\mathsf{B}^1_\dlS(a_\dlM)$ and $\dlM \not\models 
\mathsf{B}^2_\dlS(a_\dlM)$, where the relation $\models$ is defined 
over \dlLiteF concept expressions as follows: 
\begin{align*}
\dlM \models& \bot(a_\dlM) \quad \leftrightarrow \quad \textit{false} 
\\ 
\dlM \models& \mathsf{A}_\dlS(a_\dlM) \quad \leftrightarrow \quad 
\mathsf{A}_\dlS(a_\dlM) 
\!\in\! 
\dlM \\ 
\dlM \models& \exists\mathsf{r}_\dlS(a_\dlM) \quad \leftrightarrow 
\quad \exists b.\ 
\mathsf{r}_\dlS(a_\dlM,b) 
\!\in\! 
\dlM \\ 
\dlM \models& \exists\mathsf{r}^-_\dlS(a_\dlM) \quad \leftrightarrow 
\quad \exists b.\ 
\mathsf{r}_\dlS(b,a_\dlM) 
\!\in\!\dlM .
\end{align*}
The knowledge base is \emph{core-complete} if such a node cannot be 
found.
\end{proof}

\section{Actions}\label{sec:act}

We now introduce a formal language to encode mutating actions. 
Let us remind ourselves that, in our application of interest, the 
execution of a mutating action modifies the configuration of a 
deployment by either adding new resource instances, deleting 
existing ones, or modifying their settings. Here, we introduce~a 
framework for \dlLiteF core-closed knowledge base updates, triggered 
by the execution of an action that enables all the above mentioned 
effects. The only component of the core-closed knowledge base that is 
modified by the action execution is \dlM; while \dlT, \dlS, and \dlA 
remain unchanged. As a consequence of updating \dlM, actions can 
introduce new individuals and delete old ones, thus updating the set 
\modelNodes as well. Note that this may force 
changes outside \modelNodes due to the axioms in $\dlT$ and $\dlS$. 
%
The effects of applying an action over \dlM depend on a set of input 
parameters that will be instantiated at execution time, resulting in 
different assertions being added or removed from \dlM. 
As a consequence of assertions being added, fresh individuals might 
be introduced in the active domain of \dlM, including both model 
nodes from \modelNodes and boundary nodes from \boundaryNodes. 
Differently, as a consequence of assertions being 
removed, individuals might be removed from the active domain of \dlM, 
including model nodes from \modelNodes but \emph{not} 
including boundary nodes from \boundaryNodes. In fact, 
boundary nodes are owned by the open portion of the 
knowledge base and are known to exist regardless of them being used 
in \dlM. We invite the reader to 
review the set definitions for \dlA- and \dlM-assertions 
(\cref{sec:background}) to note that it is indeed possible for a 
generic  boundary individual $a$ involved in an \dlM-assertion to 
also be involved in an \dlA-assertion.

\subsection{Syntax}

An action is defined by a signature and a body. The signature 
consists of an action name and a list of formal parameters, which 
will be replaced with actual parameters at execution time. The 
body, or action effect, can include conditional statements and 
concatenation of atomic operations over \dlM-assertions.  
For example, let $\alpha$ be the action $act (\vec{x}) = \gamma$; 
that is, the action denoted by signature $act(\vec{x})$ and body 
$\gamma$, with signature name $act$, signature parameters $\vec{x}$, 
and body effect $\gamma$. Since it contains unbound parameters, or 
free variables, action $\alpha$ is ungrounded and needs to be 
instantiated with actual values in order to be executed over a set 
\dlM. In the following, we assume the existence of a set 
$\mathsf{Var}$, of variable names, and consider a generic input 
parameters substitution $\vec\theta : \mathsf{Var} \rightarrow 
\individuals$, which replaces each variable name by an individual 
node. For simplicity, we will denote an ungrounded action 
by its effect $\gamma$, and a grounded action by the composition of 
its effect with an input parameter substitution $\gamma\vec{\theta}$.
Action effects can either be \emph{complex} or \emph{basic}. 
The syntax of complex action effects $\gamma$ and basic 
effects $\beta$ is constrained by the following grammar.
\begin{align*}
\gamma ::=\ & 
	\epsilon\ \mid\  
	\beta \cdot \gamma\ \mid\  
	[\ \!\varphi \rightsquigarrow \beta\ \!] \cdot \gamma\\
\beta ::=\ & 
	\oplus_x S \ \mid\   
	\ominus_x\ S \ \mid\  
	\odot_{x_{new}} S \ \mid\ 
	\ominus_x 
\end{align*}

The complex action effects $\gamma$ include: the empty effect 
($\ \!\epsilon\ \!$), the execution of a basic effect
followed by a complex one ($\ \!\beta\cdot\gamma\ \!$), and the 
conditional execution of a basic effect upon evaluation of a formula 
$\varphi$ over the set \dlM ($\ \![\ \!\varphi \rightsquigarrow 
\beta\ \!] \cdot \gamma\ \!$).
The basic action effects $\beta$ include: the addition of a set $S$ 
of \dlM-assertions to the subset $\dlM_x$ ($\ \!\oplus_x S\ \!$), the 
removal of a set $S$ of \dlM-assertions from the subset $\dlM_x$
($\ \!\ominus_x S\ \!$), the addition of a fresh subset 
$\dlM_{x_{new}}$ containing all the \dlM-assertions in the set $S$ 
($\ \!\odot_{x_{new}} S\ \!$), and the removal of an existing 
$\dlM_x$ subset in its entirety ($\ \!\ominus_x\ \!$).
The set $S$, the formula $\varphi$, and the operators 
$\oplus/\ominus$ might contain \emph{free variables}. 
These variables are of two types: \emph{(1)} variables that are 
replaced by the grounding of the action input parameters, and  
\emph{(2)} variables that are the answer variables of the formula 
$\varphi$ and appear in the nested effect $\beta$.

\begin{example}
The following is the definition of the action \textsf{createBucket} 
from the API reference of the AWS resource type 
$\mathsf{S3\!\!::\!\!Bucket}$. The input parameters 
are two: the new bucket name $``name"$ and the canned access 
control list $``acl"$ (one of \emph{Private}, \emph{PublicRead}, 
\emph{PublicReadWrite}, \emph{AuthenticatedRead}, etc.). The effect 
of the action is to add a fresh subset $\dlM_x$ for the newly 
introduced individual $x$ containing the two assertions 
$\mathsf{S3\!\!::\!\!Bucket}(x)$ and $\mathsf{accessControl}(x,y)$. 
\begin{align*}
\mathsf{createBucket}(x :name, y :acl) = \odot_{x} \{\  
\mathsf{S3\!\!::\!\!Bucket}(x), \mathsf{accessControl}(x,y)\ \} \cdot 
\epsilon
\end{align*}
The action needs to be instantiated by a specific parameter 
assignment, for example the substitution $\theta = 
[\ x \leftarrow DataBucket,\ y \leftarrow Private\ ]$, which binds 
the variable $x$ to the node $DataBucket$ and the variable $y$ to the 
node $Private$, both taken from a pool of inactive nodes in 
\individuals. 
\end{example}
\paragraph{Action Query $\varphi$}
The syntax introduced in the previous paragraph allows for complex 
actions that conditionally execute a basic effect $\beta$ depending 
on the evaluation of a formula $\varphi$ over \dlM. This is done via 
the construct $[\ 
\!\varphi \rightsquigarrow \beta\ \!]\cdot \gamma$. The formula 
$\varphi$ might have a set $\vec{y}$ of answer variables that appear 
free in its body and are then bound to concrete tuples 
of nodes during evaluation. The answer tuples are in turn used to 
instantiate the free variables in the nested effect $\beta$.
We call $\varphi$ the \emph{action query} since we use it to select 
all the nodes that will be involved in the action effect. According 
to the grammar below, $\varphi$ is a boolean combination of 
\dlM-assertions potentially containing free variables. 
\begin{align*}
\varphi ::=\ & 
	\mathsf{A}_\dlS(t) \ \mid\  
	\mathsf{R}_\dlS(t_1,t_2) \mid\  
	\varphi_1 \wedge \varphi_2 \ \mid\ 
	\varphi_2 \vee \varphi_2\ \mid\ 
	\neg \varphi
\end{align*}
In particular, $\mathsf{A}_\dlS$ is a symbol from the set 
\specConcepts of partially-closed concepts; 
$\mathsf{R}_\dlS$ is a symbol from the set \specRoles of 
partially-closed roles; and $t,t_1,t_2$ are either individual or 
variable names from the set $\individuals \uplus \mathsf{Var}$, chosen
in such a way that the resulting assertion is an \dlM-assertion. 
Since the formula $\varphi$ can only refer to \dlM-assertions, which 
are interpreted under a closed semantics, its evaluation requires 
looking at the content of the set \dlM. 
A formula $\varphi$ with no free variables is a boolean formula and 
evaluates to either true or false.
A formula $\varphi$ with answer variables $\vec{y}$ and arity 
$ar(\varphi)$ evaluates to all the tuples $\vec{t}$, of size 
equal the arity of $\varphi$, that make the formula true in \dlM. 
The free variables of $\varphi$ can only appear in the action 
$\beta$ such that $\varphi \rightsquigarrow \beta$.
We denote by $\mathsf{ANS}(\varphi,\dlM)$ the set of answers to the 
action query $\varphi$ over \dlM. It is easy to see that the maximum 
number of tuples that could be returned by the evaluation (that is, 
the size of the set $\mathsf{ANS}(\varphi,\dlM)$) is bounded by 
$|\modelNodes \uplus \boundaryNodes|^{ar(\varphi)}$, in turn bounded 
by $(\ \!2 |\dlM|\ \!)^{2|\varphi|}$.

\begin{example}
The following example shows the encoding of the S3 API operation 
called $\mathsf{deleteBucketEncryption}$, which requires as unique 
input parameter the name of the bucket whose encryption configuration 
is to be deleted. Since a bucket can have multiple encryption 
configuration rules (each prescribing different encryption keys and 
algorithms to be used) we use an action query $\varphi$ to select 
\emph{all} the nodes that match the assertions structure to be 
removed.  
\[
\varphi[y,k,z](x) = 
\mathsf{S3\!\!::\!\!Bucket}(x) \wedge \mathsf{encrRule}(x,y) 
\wedge \mathsf{SSEKey}(y,k) \wedge \mathsf{SSEAlgo}(y,z)
\]
The query $\varphi$ is instantiated by the specific bucket instance 
(which will replace the variable $x$) and returns all the triples 
$(y,k,z)$ of encryption rule, key, and algorithm, respectively, which 
identify the assertions corresponding to the different encryption 
configurations that the bucket has. The answer variables are 
then used in the action effect to instantiate the assertions to 
remove from $\dlM_{x}$: 
\begin{align*}
&\mathsf{deleteBucketEncryption}(x : name) = \\
&[\ \ 
\varphi[y,k,z](x)\ \rightsquigarrow\ \ominus_{x} \{\ 
\mathsf{encrRule}(x,y),
\mathsf{SSEKey}(y,k),
\mathsf{SSEAlgo}(y,z)
\ \}
\ \ ] \cdot \epsilon
\end{align*}
\end{example}

\subsection{Semantics}
So far, we have described the syntax of our action language and 
provided two examples that showcase the encoding of real-world API 
calls. Now, we define the semantics of action effects with respect to 
the changes that they induce over a knowledge base.
Let us recall that given a substitution $\vec\theta$ for the input 
parameters of an action $\gamma$, we denote by 
$\gamma\vec\theta$ the grounded action where all the input variables 
are replaced according to what prescribed by $\vec\theta$. Let us 
also recall that the effects of an action apply only to assertions in 
\dlM and individuals from \modelNodes, and cannot 
affect nodes and assertions from the open portion of the knowledge 
base. 

The execution of a grounded action $\gamma\vec\theta$ over a 
\dlLiteF core-closed knowledge base $\dlK = (\dlT, \dlA, \dlS, 
\dlM)$, defined over the set \modelNodes of partially-closed 
individuals, generates a 
new knowledge base $\dlK^{\gamma\vec\theta} = (\dlT, \dlA,\dlS, 
\dlM^{\gamma\vec\theta})$, defined over 
an updated set of partially-closed individuals 
$\mathbf{I}^{\dlM^{\gamma\vec\theta}}$.
Let $S$ be a set of \dlM-assertions, $\gamma$ a complex action, 
$\vec\theta$ an input parameter substitution, and $\vec\rho$ 
a generic substitution that potentially replaces all 
free variables in the action $\gamma$. 
Let $\vec\rho_1$ and $\vec\rho_2$ be two substitutions with signature 
$\mathsf{Var} \rightarrow \individuals$ such that $dom(\vec\rho_1) 
\cap dom(\vec\rho_2) = \emptyset$; we denote their composition by 
$\vec\rho_1 \vec\rho_2$ and define it as the new substitution such 
that $\vec\rho_1\vec\rho_2(x) = a$ if $\vec\rho_1(x)\!=\!a\ \vee\ 
\vec\rho_2(x)\!=\!a$, and $\vec\rho_1\vec\rho_2(x) = \bot$ if 
$\vec\rho_1(x)\!=\!\bot\ \wedge\ 
\vec\rho_2(x)\!=\!\bot$.
We formalize the application of the grounded action 
$\gamma\vec\theta$ as the transformation $T_{\gamma\vec\theta}$ that 
maps the pair $\left< \dlM,\modelNodes \right>$ into the new pair 
$\left< \dlM',{\mathbf{I}^{\dlM}}' \right>$. We sometimes use the 
notation $T_{\gamma\vec\theta}(\dlM)$ or 
$T_{\gamma\vec\theta}(\modelNodes)$ to refer to the updated 
MBox or to the updated set of model nodes, respectively.
The rules for applying the transformation depend on the structure of 
the action $\gamma$ and are reported below. 
The transformation starts with an initial generic substitution 
$\vec\rho = \vec\theta$. As the transformation progresses, the 
generic substitution $\vec\rho$ can be updated only as a result of 
the evaluation of an action query $\varphi$ over \dlM. Precisely, all 
the tuples $\vec{t_1},...,\vec{t_n}$ making $\varphi$ true in \dlM 
will be considered and composed with the current substitution 
$\vec\rho$ generating $n$ fresh substitutions $\vec{\rho 
t_1},...,\vec{\rho t_n}$ which are used in the subsequent 
application of the nested effect $\beta$.
\begin{align*}
T_{\epsilon\vec\rho}
(\dlM,\modelNodes)
=& (\dlM,\modelNodes)\\
T_{\beta\cdot\gamma\vec\rho}
(\dlM,\modelNodes)
=& 
T_{\gamma\vec\rho}\big(\ 
T_{\beta\vec\rho}(\dlM,\modelNodes)\ \big)\\
T_{[ \varphi \rightsquigarrow \beta] \cdot 
\gamma\vec\rho}
(\dlM,\modelNodes)
=& 
\begin{cases}
T_{\gamma\vec\rho} \big(
T_{\beta\vec\rho}(\dlM,\modelNodes)\ \big) 
\qquad\quad\  
\text{if 
}\mathsf{ANS}(\varphi,\dlM)= tt
\\
T_{\gamma\vec\rho}(\dlM,\modelNodes)
\qquad\qquad\qquad\ 
\text{if 
}\mathsf{ANS}(\varphi,\dlM)=\emptyset \textit{ or } {f\!\!f}
\\ 
T_{\gamma\vec\rho} \big(T_{\beta{\vec\rho\vec{t}_1} \cdot .. \cdot 
\beta{\vec\rho\vec{t}_n}}(\dlM,\modelNodes)\big) 
\ 
\text{if 
}\mathsf{ANS}(\varphi,\dlM)=\{ \vec{t}_1,..,\vec{t}_n \}\\ 
\end{cases} 
\\ 
T_{\oplus_x S\vec\rho}
(\dlM,\modelNodes)
=& 
\big(\ \{\dlM_i\}_{i\not= 
\vec\rho(x)} 
\cup 
\{\dlM_{\vec\rho(x)}\cup 
S_{\vec\rho}\}\ ,\ \individuals^\dlM\cup ind(S_{\vec\rho})\ \big) \\
T_{\ominus_x S\vec\rho }
(\dlM,\modelNodes)
=& 
\big(\ \{\dlM_i\}_{i\not= 
\vec\rho(x)} 
\cup 
\{\dlM_{\vec\rho(x)}\smallsetminus 
S_{\vec\rho}\}\ ,\ \individuals^\dlM\smallsetminus 
ind(S_{\vec\rho})\ \big) \\
T_{\odot_{x} S\vec\rho}
(\dlM,\modelNodes)
=& 
\big(\  
\dlM 
\cup \{
\dlM_{\vec\rho(x)}= S_{\vec\rho}\}\ ,\ \individuals^\dlM\cup 
ind(S_{\vec\rho})\ 
\big)\\
T_{\ominus_x\vec\rho} 
(\dlM,\modelNodes)
=& \big(\ \dlM 
\smallsetminus 
\dlM_{\vec\rho(x)}\ ,\ \modelNodes\smallsetminus 
ind(\dlM_{\vec\rho(x)})\ 
\big)\\
\end{align*}
Since the core \dlM of the knowledge base \dlK changes at every 
action execution, its domain of model nodes \modelNodes changes as 
well. The execution of an action $\gamma\vec\theta$ over the 
knowledge base $\dlK = (\dlT, \dlA,\dlS, \dlM)$ with set of model 
nodes \modelNodes could generate a new $\dlK^{\gamma\vec\theta} = 
(\dlT, \dlA,\dlS, \dlM^{\gamma\vec\theta})$ with a new set of model 
nodes ${\modelNodes}'$ that is not \emph{core-complete} or 
not \emph{open-consistent} (see~\cref{sec:coreComp} for the 
corresponding 
definitions).
We illustrate two examples next. 

\begin{example}[Violation of core-completeness]
Consider the case where the general specifications of the system 
require all objects of type bucket to have a logging configuration, 
and an action that removes the logging configuration from a bucket. 
Consider the core-closed knowledge base \dlK where 
$\dlS = \{ \mathsf{S3\!\!::\!\!Bucket} \sqsubseteq \exists 
\mathsf{loggingConfiguration}\}$ and 
$\dlM = \{\mathsf{S3\!\!::\!\!Bucket}(b), 
\mathsf{loggingConfiguration}(b,c)\}$ (consistent wrt \dlS) and the 
action $\gamma$ defined as 
\begin{align*}
&\mathsf{deleteLoggingConfiguration}(x: name) =\\ 
&\qquad[\ (\varphi[y](x) = \mathsf{S3\!\!::\!\!Bucket}(x) \wedge 
\mathsf{loggingConfiguration}(x,y)) \\
&\qquad\qquad\rightsquigarrow 
\ominus_{x}\{\mathsf{loggingConfiguration}(x,y)\}\ ] \cdot \epsilon
\end{align*}
For the input parameter substitution $\vec\theta = [\ x\leftarrow b 
\ ]$, it is easy to see that the transformation 
$T_{\gamma\vec\theta}$ 
applied to \dlM results in the update $\dlM^{\gamma\vec\theta}= \{
\mathsf{S3\!\!::\!\!Bucket}(b)
\}$, which is \emph{not} core-complete. 
\end{example}

\begin{example}[Violation of open-consistency]
Consider the case where an action application indirectly affects 
boundary nodes and their properties, leading to inconsistencies in 
the open portion of the knowledge base. For example, when the 
knowledge base prescribes that buckets used to store logs cannot be 
public; however, a change in the configuration of a bucket instance 
causes a second bucket (initially known to be public) to also 
become a log store. In particular, this happens when the knowledge 
base \dlK contains the \dlT-axiom $\exists 
\textsf{loggingDestination}^- \! \sqsubseteq \! 
\neg\mathsf{PublicBucket}$ and the \dlA-assertion 
$\mathsf{PublicBucket}(b)$, and we apply an action that 
introduces a new bucket storing its logs to $b$, defined as follows: 
\begin{align*}
&\mathsf{createBucketWithLogging}(x: name, y: log) = \\
&\qquad \odot_{x} \{ 
\mathsf{S3\!\!::\!\!Bucket}(x), \mathsf{loggingDestination}(x,y) \} 
\end{align*}
For the input parameter substitution $\vec\theta = [\ x \leftarrow 
newBucket, y \leftarrow b\ ]$, the result of applying the 
transformation $T_{\gamma\vec\theta}$ is the set $\dlM =\{\  
\mathsf{S3\!\!::\!\!Bucket}(newBucket),$ $ 
\mathsf{loggingDestination}(newBucket,b)\ \}$ which, combined with 
the pre-existing and unchanged sets \dlT and \dlA, causes the updated 
$\dlK^{\gamma\vec\theta}$ to be \emph{not} open-consistent. 
\end{example}

From a practical point of view, the examples highlight the need to 
re-evaluate core-completeness and open-consistency of a core-closed 
knowledge base after each action execution. Detecting a violation to 
core-completeness signals that we have modeled an action that is 
inconsistent with respect to the systems specifications, which most 
likely means that the action is missing something and needs to be 
revised. Detecting a violation to open-consistency signals that our 
action, even when consistent with respect to the specifications, 
introduces a change that conflicts with other assumptions that we 
made about the system, and generally indicates that we should either 
revise the assumptions or forbid the application of the action. Both 
cases are important to consider in the development life cycle of the  
core-closed KB and the action definitions.

\section{Static Verification}\label{sec:sv}

In this section, we investigate the problem of computing whether 
the execution of an action, no matter the specific instantiation, 
always preserves given properties of core-closed knowledge bases. We 
focus on properties expressed as \textsc{Must}/\textsc{May} queries 
and define the static verification problem as follows.

\begin{definition}[Static Verification]
Let \dlK be a \dlLiteF core-closed knowledge base, $q$ be a 
\textsc{Must}/\textsc{May} query, and $\gamma$ be an action with free 
variables from the language presented above. Let $\vec\theta$ be an 
assignment for the input variables of $\gamma$ that transforms $\gamma$ into the \emph{grounded} 
action 
$\gamma\vec\theta$. Let $\dlK^{\gamma\vec\theta}$ be the \dlLiteF 
core-closed 
knowledge base resulting from the application of the grounded action 
$\gamma\vec\theta$ onto $\dlK$. We say that the action $\gamma$ 
\emph{``preserves $q$ over $\dlK$"} iff for every grounded instance 
$\gamma\vec\theta$ we have that $\mathsf{ANS}(q,\dlK)=\mathsf{ANS} 
(q,\dlK^{\gamma\vec\theta})$. The static verification problem is that 
of 
determining whether an action $\gamma$ is $q$-preserving over $\dlK$.
\end{definition}

An action $\gamma$ is \emph{not} $q$-preserving over $\dlK$ iff there 
exists a grounding $\vec\theta$ for the input variables of $\gamma$ 
such 
that $\mathsf{ANS}(q,\dlK) \not= \mathsf{ANS} 
(q,\dlK^{\gamma\vec\theta})$; that is, fixed the grounding 
$\vec\theta$ there 
exists a tuple $\vec{t}$ for $q$'s answer variables such that 
$\vec{t}\in\mathsf{ANS}(q,\dlK) \smallsetminus 
\mathsf{ANS}(q,\dlK^{\gamma\vec\theta})$ or 
$\vec{t}\in\mathsf{ANS}(q,\dlK^{\gamma\vec\theta}) \smallsetminus 
\mathsf{ANS}(q,\dlK)$.

\begin{theorem}[Complexity of the Static Verification Problem]
The static verification problem, i.e.deciding whether an action 
$\gamma$ is $q$-preserving over \dlK, can be decided in 
polynomial time in data complexity.
\end{theorem}

\begin{proof}
The proof relies on the fact that one could: enumerate all possible 
assignments $\vec\theta$; compute the updated knowledge bases 
$\dlK^{\gamma\vec\theta}$; check whether these are fully satisfiable; 
enumerate all tuples $\vec t$ for the query $q$; and, finally, check 
whether there exists at least one such tuple that satisfies $q$ over 
$\dlK$ but not $\dlK^{\gamma\vec\theta}$ or vice versa.
The number of assignments $\vec\theta$ is bounded by 
$\big(|\modelNodes \uplus \openIndividuals| + 
ar(\gamma)\big)^{ar(\gamma)}$ as it is sufficient to replace each 
variable appearing in the action $\gamma$ either by a known object 
from $\modelNodes \uplus \openIndividuals$ or by a fresh one.
The computation of the updated 
$\dlK^{\gamma\vec\theta}$ is done in polynomial time in \dlM (and 
is exponential in the size of the action $\gamma$) as it may require 
the evaluation of an internal action query $\varphi$ and the 
consecutive re-application of the transformation for a number of 
tuples that is bounded by a polynomial over the size of \dlM.
As explained in~\cref{sec:coreComp}, checking full satisfiability of the 
resulting core-closed knowledge base is also polynomial in \dlM.
The number of tuples $\vec t$ is bounded by $\big(|\modelNodes \uplus 
\openIndividuals| + ar(\gamma)\big)^{ar(q)}$ as it is enough to 
consider all those tuples involving known objects plus the fresh 
individuals introduced by the assignment $\vec\theta$.
Checking whether a tuple $\vec t$ satisfies the query $q$ over a 
core-closed knowledge base is decided in $\textsc{LogSpace}$ in the 
size of \dlM~\cite{CauliOP21} which is, thus, also 
polynomial in \dlM.
\end{proof}

\section{Planning}\label{sec:ts}

As discussed throughout the paper, the execution of a mutating action 
modifies the configuration of a deployment and potentially changes 
its posture with respect to a given set of requirements. In the 
previous two sections, we introduced a language to encode mutating 
actions and we investigated the problem of checking whether the 
application of an action preserves the properties of a core-closed 
knowledge base. In this section, we investigate the plan existence 
and synthesis problems; that is, the problem of deciding whether 
there exists a sequence of grounded actions that leads the knowledge 
base to a state where a certain requirement is met, and the problem 
of finding a set of such plans, respectively. We start by defining a 
notion of transition system that is generated by applying actions to 
a core-closed knowledge base and then use this notion to focus on the 
mentioned planning problems.
As in classical planning, the plan existence problem for plans 
computed over unbounded domains is 
undecidable~\cite{ErolNS95,Chapman87}. The undecidability proof is 
done via reduction from the Word problem.
The problem of deciding whether a deterministic Turing machine $M$ 
accepts a word $w\in\{0,1\}^*$ is reduced to the plan existence 
problem. Since undecidability holds even for basic action effects, we 
can show undecidability over an unbounded domain by using the same 
encoding of~\cite{AhmetajCOS17}.

\paragraph{Transition Systems}
In the style of the work done in~\cite{CalvaneseMPG15,HaririCMGMF13}, 
the combination of a \dlLiteF core-closed knowledge base and a set of 
actions can be viewed as the transition system it generates.
Intuitively, the states of the transition system correspond to MBoxes 
and the transitions between states are labeled by grounded actions. 
%
%
A \dlLiteF core-closed knowledge base $\dlK = (\dlT, \dlA, 
\dlS, \dlM_0)$, defined over the possibly infinite set of individuals 
$\individuals$ (and model nodes $\modelNodes_0 \subseteq \individuals 
$) and the set \textsf{Act} of ungrounded actions, generates the 
transition system (TS) $\Upsilon_\dlK = (\individuals, \dlT, \dlA, 
\dlS, \Sigma, \dlM_0, \rightarrow)$ where $\Sigma$ is a set of 
\emph{fully satisfiable} (i.e., \emph{core-complete} and 
\emph{open-consistent}) MBoxes; $\dlM_0$ is the initial MBox; and 
$\rightarrow \subseteq 
\Sigma \times L_{\mathsf{Act}} \times \Sigma$ is a labeled transition 
relation with $L_\mathsf{Act}$ the set of all possible \emph{grounded 
actions}. The sets $\Sigma$ and $\rightarrow$ are defined by mutual 
induction as the smallest sets such that: if $\dlM_i \in \Sigma$ 
then for every grounded action $\gamma\vec\theta \in L_\mathsf{Act}$ 
such that the fresh MBox $\dlM_{i+1}$ resulting from the 
transformation  $T_{\gamma\vec\theta}$ is core-complete and 
open-consistent, we have that $\dlM_{i+1} \in \Sigma$ and $(\dlM_i, 
\gamma\vec\theta, \dlM_{i+1}) \in \rightarrow$.

Since we assume that actions have input parameters that are replaced 
during execution by values from \individuals, which contains both 
known objects from $\modelNodes \uplus \openIndividuals$ and possibly 
infinitely many fresh objects, the generated transition system 
$\Upsilon_\dlK$ is generally infinite. 
To keep the planning problem decidable, we concentrate on a known 
finite subset $\dl{D}\subset\individuals$ containing all the fresh 
nodes and value assignments to action variables that are of interest 
for our application. In the remainder of this paper, we discuss the 
plan existence and synthesis problem for finite transition systems 
$\Upsilon_\dlK = (\dl{D}, \dlT, \dlA, \dlS, \Sigma, \dlM_0, 
\rightarrow)$, whose states in $\Sigma$ have a domain that is also 
bounded by $\dl D$.

\paragraph{The Plan Existence Problem}
A {plan} is a sequence of grounded actions whose execution leads to a 
state satisfying a given property. Let $\dlK = (\dlT, \dlA, \dlS, 
\dlM_0)$ be a \dlLiteF core-closed knowledge base; \textsf{Act} be a 
set of ungrounded actions; and let $\Upsilon_\dlK = (\dl D, 
\dlT, \dlA,\dlS,\Sigma,\dlM_0,\rightarrow)$ be its generated finite 
TS. Let 
$\pi$ be a finite 
sequence $\gamma_1\vec\theta_1\cdots\gamma_n\vec\theta_n$ of grounded 
actions taken from the set $L_\mathsf{Act}$. We call the sequence 
$\pi$ \emph{consistent} iff there exists a run $\rho = \dlM_0 
\xrightarrow{\gamma_1\vec\theta_1} \dlM_1 
\xrightarrow{\gamma_2\vec\theta_2} \cdots 
\xrightarrow{\gamma_n\vec\theta_n} \dlM_n$ in $\Upsilon_\dlK$. 
Let $q$ be a \textsc{Must}/\textsc{May} query mentioning objects from 
$adom(\dlK)$ and $\vec t$ a tuple from the set 
$adom(\dlK)^{ar(q)}$. A consistent sequence $\pi$ of grounded 
actions is a \emph{plan} from $\dlK$ to $(\vec t, q)$ iff $\vec t \in 
\mathsf{ANS}(q, \dlK_n = (\dlT,\dlA,\dlS,\dlM_n))$ with $\dlM_n$ the 
final state of the run induced by $\pi$.
\begin{definition}[Plan Existence]
Given a \dlLiteF core-closed knowledge base $\dlK$, a tuple $\vec t$, 
and a \textsc{Must}/\textsc{May} query $q$, the \emph{plan existence} 
problem is that of deciding whether there exists a plan from \dlK to 
$(\vec t, q)$.
\end{definition}

\begin{example}
Let us consider the transition system $\Upsilon_\dlK$ generated by 
the core-closed knowledge base $\dlK = (\dlT,\dlA,\dlS,\dlM_0)$ 
having 
the set of partially-closed assertions $\dlM_0$ defined as 
\begin{align*}
\{\  
&\mathsf{S3\!\!::\!\!Bucket}(b),\  
\mathsf{KMS\!\!::\!\!Key}(k),\
\mathsf{bucketEncryptionRule}(b,r),\  
\mathsf{bucketKey}(r,k), \\
&\mathsf{bucketKeyEnabled}(r,true),\  
\mathsf{enableKeyRotation}(k,false)
\ \}
\end{align*}
and the set of action labels \textsf{Act} containing the actions 
\textsf{deleteBucket}, 
\textsf{createBucket}, 
\textsf{deleteKey}, 
\textsf{createKey}, 
\textsf{enableKeyRotation}, 
\textsf{putBucketEncryption}, and 
\textsf{deleteBucketEncryption}.
Let us assume that we are interested in verifying the existence of a 
sequence of grounded actions that when applied onto the knowledge 
base would configure the bucket node $b$ to be encrypted with a 
rotating key. 
Formally, this is equivalent to checking the existence of a 
consistent plan $\pi$ that when executed on the transition system 
$\Upsilon_\dlK$ leads to a state $\dlM_n$ such that the tuple $\vec t 
= b$ is in the set $\mathsf{ANS}(q, \dlK_n = 
(\dlT,\dlA,\dlS,\dlM_n))$ for $q$ the query
\begin{align*}
q[x] =\ 
&\mathsf{S3\!::\!Bucket}(x)\ \wedge\ \textsc{Must}\ \big( \ 
\exists  y,z .\ \mathsf{bucketSSEncryption}(x,y)\ \wedge\ \\
&\mathsf{bucketKey}(y,z) 
\ \wedge\ \mathsf{enableKeyRotation}(z, true)\ \big)
\end{align*}
It is easy to see that the following three sequences of grounded 
actions are valid plans from \dlK to $(b,q)$:
\begin{align*}
\pi_1 = \ &\mathsf{enableKeyRotation}(k)\ \\
\pi_2 = \ &\mathsf{createKey}(k_1) \cdot 
\mathsf{enableKeyRotation}(k_1) \cdot 
\mathsf{putBucketEncryption}(b,k_1)\ \\
\pi_3 = \ 
&\mathsf{deleteBucketEncryption}(b,k)\cdot
\mathsf{createKey}(k_1) \cdot 
\mathsf{enableKeyRotation}(k_1) \cdot \\
&\ \mathsf{putBucketEncryption}(b,k_1)\ 
\end{align*}
If, for example, a bucket was only allowed to have one encryption (by 
means of a functional axiom in \dlS), then $\pi_2$ would not be a 
valid plan, as it would generate an inconsistent run leading to a 
state $\dlM_i$ that is not open-consistent w.r.t. \dlS.
\end{example}

\begin{lemma}
The plan existence problem for a finite transition system 
$\Upsilon_\dlK$ generated by a \dlLiteF core-closed knowledge base 
$\dlK$ and a set of actions $\mathsf{Act}$, over a finite domain of 
objects \dl D, reduces to graph reachability over a graph whose 
number of 
states is at most exponential in the size of \dl D.
\end{lemma}

\paragraph{The Plan Synthesis Problem}
We now focus on the problem of finding plans that satisfy a given 
condition. As discussed in the previous paragraph, we are mostly 
driven by query answering; in particular, by conditions corresponding 
to a tuple (of objects from our starting deployment configuration) 
satisfying a given requirement expressed as a 
$\textsc{Must}/\textsc{May}$ query. 
Clearly, this problem is meaningful in our application of interest 
because it corresponds to finding a set of potential sequences of 
changes that would allow one to reach a configuration satisfying 
(resp., not satisfying) one, 
or more, security mitigations (resp., vulnerabilities).
We concentrate on \dlLiteF core-closed knowledge bases and their 
generated finite transition systems, where potential fresh objects 
are drawn from a fixed set \dl D.
%
We are interested in sequences of grounded actions that are minimal 
and ignore sequences that extend these. We sometimes call such 
mimimal sequences \emph{simple plans}. A plan $\pi$ from an initial 
core-closed knowledge base $\dlK$ to a goal condition $b$ is minimal 
(or simple) \emph{iff} 
there does not exist a plan $\pi'$ (from the same initial \dlK 
to the same goal condition $b$) s.t. $\pi = \pi' \cdot 
\sigma$, for $\sigma$ a non-empty suffix of grounded actions. 

In~\cref{algo:findPlans}, we present a depth-first search 
algorithm that, starting from \dlK, searches for all simple plans 
that achieve a given target query membership condition. The 
transition system 
$\Upsilon_\dlK$ is computed, and stored, on the fly in the 
\textsf{Successors} sub-procedure and the graph is explored in a 
depth-first search traversal fashion.
%
We note that the condition $\vec t \in \mathsf{ANS}(q, \left< \dlT, 
\dlA, \dlS, \dlM \right>)$ (line 9) could be replaced by any other 
query satisfiability condition and that one could easily rewrite the 
algorithm to be parameterized by a more general boolean goal. 
For example, the condition that a given tuple $\vec t$ is \emph{not} 
an answer to a query $q$ over the analyzed state, with the query $q$ 
representing an undesired configuration, or a boolean formula over 
multiple query membership assertions.
We also note that~\cref{algo:findPlans} could be simplified to return 
only one simple plan, if a plan exists, or \textsf{NULL}, if a plan 
does not exist, thus solving the so-called \emph{plan generation 
problem}. We include the modified algorithm 
in~\cref{appendix:findPlanAlgo} and the proofs of the following 
Theorems 
in~\Cref{proof:findPlansCorrectness,proof:findPlansComplexity}.

\setcounter{AlgoLine}{0}
\begin{algorithm}
	\smallskip
	\SetKwProg{Fn}{def}{\string:}{}
	\SetKwFunction{FunPSyAll}{FindPlans}%
	\SetKwFunction{FunPSeAll}{AllPlanSearch}%
	\SetKwFunction{FunSucc}{Successors}%
	\SetKwInOut{Inputs}{Inputs}\SetKwInOut{Output}{Output}
\Inputs{A ccKB $\dlK = (\dlT, \dlA, \dlS, \dlM_0)$, a domain \dl D, a 
set of actions $\mathsf{Act}$ and a pair $\left<\vec{t},q\right>$ of 
an answer tuple and a \textsc{Must}/\textsc{May} query}
\Output{An possibly empty set $\Pi$ of consistent 
simple plans}
\BlankLine

\Fn{\FunPSyAll($\ 
\dlK, 
\dl D,
\mathsf{Act},
\left< \vec{t},q \right>$)}
{
$\Pi := \emptyset$\;
$S := \bot$\;
{ $\mathsf{AllPlanSearch}(
\dlM_0, \epsilon, \emptyset, 
\dlK,
\dl D,
\mathsf{Act},
\left< \vec{t},q \right>)$ }\;
\Return{$\Pi$\;}
}

\Fn{\FunPSeAll(
$ \dlM, \pi, V, \dlK,
\dl D,
\mathsf{Act},
\left< \vec{t},q \right>$)}
{
\If{ $\dlM \in V$ }{\Return\;}
\If{ $\vec{t} \in \mathsf{ANS}(q, \left< \dlT, \dlA, \dlS, \dlM 
\right>)$ }{
$\Pi := \Pi \cup \{\pi\}$\;
\Return\;
}
$Q := \emptyset$\;
\ForEach{$\left<\gamma\vec\theta, \dlM'\right> \in 
\mathsf{Successors}(\dlM,\mathsf{Act},\dl D)$}{
$Q.push(\left<\gamma\vec\theta, \dlM'\right>)$\;
}
$V := V \cup \{\dlM\}$\;
\While{ $Q \not= \emptyset$ }{
$\left< \gamma\vec\theta, \dlM'\right> = Q.pop()$\;
$\mathsf{AllPlanSearch}( \dlM', 
\pi\cdot{\gamma\vec\theta}, 
V, 
\dlK,
\dl D,
\mathsf{Act},
\left< \vec{t},q \right>) $\;
}
$V := V \smallsetminus \{\dlM\}$\;
\Return\;
}

\Fn{\FunSucc($\dlM, \mathsf{Act}, \dl D $)}
{

\If{$S[\dlM] \textit{ is defined}$}{
\Return{$S[\dlM]$}\;
}

$N := \emptyset$\;
\ForEach{ $\gamma \in \mathsf{Act}$, $\vec\theta \in 
{\dl D}^{ar(\gamma)}$ }{
$\dlM' := T_{\gamma\vec\theta}(\dlM)$\;
\If{ $\dlM' \text{is fully satisfiable} $}{$N := N \cup \{\left< 
\gamma\vec\theta, \dlM'\right> \}$}
}
$S[\dlM] := N$\;
\Return{$N$\;}
}
\caption{
\label{algo:findPlans}
\textsf{FindPlans}($\dlK, \dl D, \mathsf{Act}, 
\left<\vec{t},q\right>$)}
\end{algorithm}

\begin{theorem}[Minimal Plan Synthesis Correctness]
Let \dlK be a \dlLiteF core-closed knowledge base, \dl D be a fixed 
finite 
domain, \textsf{Act} be a set of 
ungrounded action labels, and $\left< \vec t, q \right>$ be a 
goal. Then a plan $\pi$ is returned by the algorithm  
$\mathsf{FindPlans}(\dlK,\dl D,\mathsf{Act},\left< \vec t, q 
\right>)$ if and only if $\pi$ is a minimal plan from $\dlK$ to 
$\left< \vec t, q \right>$.
\end{theorem}

\begin{theorem}[Minimal Plan Synthesis Complexity]
The $\mathsf{FindPlans}$ algorithm runs in polynomial time in $\dlM$ 
and exponential time in $\dl D$.
\end{theorem}

\section{Related Work}\label{sec:rel}

The syntax of the action language that we presented in this paper is 
similar to that of~\cite{AhmetajCOS17,CalvaneseOS13,CalvaneseOS16}. 
Differently from their work, we disallow complex action effects to be 
nested inside conditional statements, and we define basic action 
effects that consist purely in the addition and deletion of concept 
and role \dlM-assertions. Thus, our actions are much less general 
than those used in their framework. The semantics of their action 
language is defined in terms of changes applied to instances, and the 
action effects are captured and encoded through a variant of 
\dl{ALCHOIQ} called $\dl{ALCHOIQ}_{br}$. In our work,  
instead, the execution of an action updates a portion of the 
core-closed knowledge base \dlK{---}the core \dlM, which is 
interpreted 
under a close-world assumption and can be seen as a partial 
assignment for the interpretations that are models of \dlK. Since we 
directly 
manipulate \dlM, the semantics of our actions is more similar to that 
of~\cite{HaririCMGMF13} and, in general, to ABox 
updates~\cite{KharlamovZC13,LiuLMW11}. Like the frameworks introduced 
in~\cite{GiacomoMR12,CalvaneseMPS16,CGMP13,CalvaneseMPG15}, our 
actions are parameterized and when combined with a core-closed 
knowledge base generate a transition system. 
In~\cite{CalvaneseMPS16}, the authors focus on a variant of 
\emph{Knowledge and Action Bases} (\cite{HaririCMGMF13}) called 
\emph{Explicit-Input KABs} (eKABs); in particular, on finite and on 
state-bounded eKABs, for which planning existence is decidable. Our 
generated transition systems are an adaptation of the work done in 
\emph{Description Logic based Dynamic Systems}, \emph{KABs}, and 
\emph{eKABs} to our setting of core-closed knowledge bases.
In~\cite{Milicic07}, the authors address decidability of the plan 
existence problem for logics that are subset of \dl{ALCOI}. Their 
action language is similar to the one presented in this paper; 
including pre-conditions, in the form of a set of ABox assertions, 
post-conditions, in the form of basic addition or removal of 
assertions, concatenation, and input parameters. 
In~\cite{CalvaneseMPS16}, the plan synthesis problem is discussed 
also for lightweight description logics. 
Relying on the FOL-reducibility of \dlLiteA, 
it is shown that plan 
synthesis over \dlLiteA can be compiled into an ADL planning problem 
\cite{Pednault94}. 
This does not seem possible in our case, as not all necessary tests 
over core-closed knowledge bases are known to be FOL-reducible. 
In \cite{CalvaneseMPG15} and~\cite{CGMP13}, the authors concentrate on verifying and 
synthesizing temporal properties expressed in a variant of 
$\mu$-calculus over description logic based dynamic systems, both 
problems are relevant in our application scenario and we will 
consider them in future works.

\section{Conclusion}\label{sec:conc}

We focused on the problem of analyzing cloud infrastructure encoded 
as description logic knowledge bases combining complete and  
incomplete information.
From a practical standpoint, we concentrated on formalizing and 
foreseeing the impact of potential changes pre-deployment.
We introduced an action language to encode mutating actions, whose 
semantics is given in terms of changes induced to the complete 
portion of the knowledge base.
We defined the static verification problem as that of deciding  
whether the execution of an action, no matter the specific parameters 
passed, always preserves a set of properties of the knowledge base. 
We characterized the complexity of the problem and provided 
procedural steps to solve it.
We then focused on three formulations of the classical AI planning 
problem; namely, plan existence, generation, and synthesis. In our 
setting, the planning problem is formulated with respect to the  
transition system arising from the combination of a core-closed 
knowledge base and a set of actions; goals are given in terms of 
one, or more, \textsc{Must}/\textsc{May} conjunctive query membership 
assertion; and plans of interest are simple sequences of 
parameterized actions.

\newpage

\bibliographystyle{splncs04}
\bibliography{refs}

\clearpage\appendix

\renewcommand{\thesubsection}{A.\arabic{subsection}}

\subsection{\textsf{FindPlan} algorithm}\label{appendix:findPlanAlgo}

\setcounter{AlgoLine}{0}
\begin{algorithm}
	\smallskip
	\SetKwProg{Fn}{def}{\string:}{}
	\SetKwFunction{FunPSy}{FindPlan}%
	\SetKwFunction{FunPSe}{PlanSearch}%
	\SetKwFunction{FunSucc}{Successors}%
	\SetKwInOut{Inputs}{Inputs}\SetKwInOut{Output}{Output}

\Inputs{A core-closed KB $\dlK = (\dlT, \dlA, \dlS, \dlM_0)$, a 
domain \dl D, a set of ungrounded action labels $\mathsf{Act}$ and a 
pair $\left<\vec{t},q\right>$ of an answer tuple and a 
\textsc{Must}/\textsc{May} query}
\Output{\textsf{NULL} if a plan does not exist, a consistent plan 
$\pi$ otherwise}

\BlankLine

\Fn{\FunPSy($\ 
\dlK, 
\dl D,
\mathsf{Act},
\left< \vec{t},q \right>$)}
{
\tcp{ $V$ and $S$ have global scope }
$V := \emptyset$\;
$S := \bot$\;
\Return{ $\mathsf{PlanSearch}(
\dlM_0, \epsilon,
\dlK,
\dl D,
\mathsf{Act},
\left< \vec{t},q \right>)$ }\;
}

\Fn{\FunPSe($\dlM, \pi,\dlK,
\dl D,
\mathsf{Act},
\left< \vec{t},q \right>$)}
{
\If{ $\dlM \in V$ }{\Return{$\mathsf{NULL}$}\;}
$V := V \cup \{\dlM\}$\;
\If{ $\vec{t} \in \mathsf{ANS}(q, \left< \dlT, \dlA, \dlS, \dlM 
\right>)$ }{\Return{$\pi$}\;}
\ForEach{$\left<\gamma\vec\theta, \dlM'\right> \in 
\mathsf{Successors}(\dlM,\mathsf{Act},\dl D)$}{
$\pi' = \mathsf{PlanSearch}(\dlM', \pi\cdot{\gamma\vec\theta}, 
\dlK,
\dl D,
\mathsf{Act},
\left< \vec{t},q \right>) $\;
\If{$\pi' \not= \mathsf{NULL}$}{\Return{$\pi'$\;}}
}
\Return{$\mathsf{NULL}$\;}
}

\Fn{\FunSucc($\dlM, \mathsf{Act}, \dl D $)}
{

\If{$S[\dlM] \textit{ is defined}$}{
\Return{$S[\dlM]$}\;
}

$N := \emptyset$\;
\ForEach{ $\gamma \in \mathsf{Act}$ }{
\ForEach{ $\vec\theta \in {\dl D}^{ar(\gamma)}$ }{
$\dlM' := T_{\gamma\vec\theta}(\dlM)$\;
\If{ $\dlM' \text{is fully satisfiable} $}{$N := N \cup \{\left< 
\gamma\vec\theta, \dlM'\right> \}$}
}
}
$S[\dlM] := N$\;
\Return{$N$\;}
}

\caption{
\label{algo:findPlan}
\textsf{FindPlan}($\dlK, \dl D, \mathsf{Act}, 
\left<\vec{t},q\right>$)}
\end{algorithm}

\subsection{Proof of \textsf{FindPlans} 
Correctness}\label{proof:findPlansCorrectness}

\begin{proof}
\begin{itemize}
\item[$\Leftarrow$] Let us assume that the sequence $\pi$ is a 
minimal plan from $\dlK$ to $\left<\vec t, q\right>$.
Since $\pi$ is a plan, then 
it is a consistent sequence of grounded actions that generates a 
run $\rho = \dlM_0 
\xrightarrow{\gamma_1\vec\theta_1} \dlM_1 
\xrightarrow{\gamma_2\vec\theta_2} \cdots 
\xrightarrow{\gamma_n\vec\theta_n} \dlM_n$ over the transition system 
$\Upsilon_\dlK$, terminating in a state $\dlM_n$ such that $\vec t 
\in \mathsf{ANS}(q, \left< \dlT, \dlA, \dlS, \dlM_n \right>)$. Since 
$\pi$ is minimal, there is no index $i<n$ in the run $\rho$ s.t. 
$\vec t \in \mathsf{ANS}(q, \left< \dlT, \dlA, \dlS, \dlM_i 
\right>)$, which also implies that the run does not have any loops. 
The algorithm starts from $\dlM_0$, and recursively explores 
all consistent runs (via the \textsf{Successors} procedure, lines 
21-30), marking the current state as visited before the recursive 
invocation (line 15), and skipping already visited states when 
re-entering the $\mathsf{AllPlanSearch}$ function (therefore avoiding 
loops within a given explored path, lines 7-8). The search space of the algorithm includes 
$\rho$ and the node $\dlM_n$ will be reached. Finally, since the 
condition over $\dlM_n$ is true (line 9), then the sequence $\pi$ 
explored so far is added to set $\Pi$ (line 10), and will ultimately 
be returned by the $\mathsf{FindPlans}$ function (line 4).
\item[$\Rightarrow$]
Let us assume that a sequence $\pi$ is returned 
by~\Cref{algo:findPlans} but that $\pi$ is not a minimal plan from 
$\dlK$ to $\left< \vec t, q \right>$. The sequence $\pi$ is either 
(1) not a valid sequence, or (2) not a consistent sequence, or (3) 
does not reach a state satisfying the goal, or (4) is not minimal. 
The algorithm starts with an empty sequence $\epsilon$ (line 4) and 
appends a new grounded action $\gamma\vec\theta$ to the sequence only 
upon recursive invocation of the $\mathsf{AllPlanSearch}$ function. 
The grounded action $\gamma\vec\theta$ is returned by the 
$\mathsf{Successors}$ function, which explores only valid grounded 
actions and returns a pair $\left< \gamma\vec\theta,\dlM'\right>$ 
only if $\dlM'$ is fully satisfiable (line 27). Since every prefix of 
$\pi$ is built by appending valid grounded actions that lead to a 
fully-satisfiable $\dlM'$ then the whole sequence $\pi$ is valid and 
consistent (contradicting \emph{(1)} and \emph{(2)}). The sequence 
$\pi$ is returned by the algorithm, therefore it was inserted in the 
set $\Pi$ at line 10 and the goal condition was satisfied 
(contradicting \emph{(3)}). Since the algorithm returns after adding 
any sequence $\pi$ to the set $\Pi$ then there cannot exists an 
extension of such sequence (contradicting \emph{(4)}). 
\end{itemize}
\end{proof}

\subsection{Proof of \textsf{FindPlans}  
Complexity}\label{proof:findPlansComplexity}
\begin{proof}
The algorithm implements a depth-first search over a graph whose 
nodes are fully satisfiable MBoxes and edges are labeled by grounded 
actions. Nodes are never visited twice along the same recursion path, 
but can be re-visited along different recursion paths. Let us assume 
that the maximum number $n$ of input 
parameters that an action can have is known; and let us also assume 
that $n \ll \dl D$ and generally independent of either \dl D or \dlM. 
Each node can have at most $|\mathsf{Act}| \cdot |\dl D|^n $ 
successors (line 25) and the recursion depth of the graph exploration 
is linear in \dl D. 
We have that the number of nodes is bounded by a function that is 
exponential in the size of the domain \dl D. The nodes are computed 
on the fly in the \textsf{Successors} sub-procedure, which checks 
whether the current MBox is fully satisfiable (check done in 
polynomial time in \dlM and in \dl D) and stores it.
The \textsf{AllPlanSearch} algorithm explores the graph looking for 
minimal paths achieving a goal assertion. There can be max 
$\big(|\mathsf{Act}| \cdot |\dl D|^n \big)^{\dl D}$ such paths (that 
do not visit the same node twice).
\end{proof}

\end{document}